\documentclass{SCIS2016}

\begin{document}

\ArticleType{RESEARCH PAPER}{}
\Year{2016}
\Month{January}
\Vol{59}
\No{1}
\DOI{xxxxxxxxxxxxxx}
\ArtNo{xxxxxx}
\ReceiveDate{}
\AcceptDate{}
\OnlineDate{}

\title{New quaternary sequences of even length with optimal auto-correlation}{New quaternary sequences of even length with optimal auto-correlation}

\author[1,4]{Wei SU}{}
\author[2,4]{Yang YANG}{yang$\_$data@qq.com}
\author[2]{ZhengChun ZHOU}{}
\author[3]{XiaoHu TANG}{}

\AuthorMark{W Su, Y Yang, Z Zhou}

\AuthorCitation{W Su, Y Yang, Z Zhou, X Tang}

\address[1]{School of Economics and Information Engineering,\\ Southwestern University of Finance and Economics, Chengdu 610074, China}
\address[2]{School of Mathematics, Southwest Jiaotong University, Chengdu  611756, China}
\address[3]{Provincial Key Lab of Information Coding and Transmission, Institute of Mobile Communications,\\
Southwest Jiaotong University, Chengdu 611756, China}
\address[4]{Science and Technology on Communication Security Laboratory,
 Chengdu 610041, China}

\abstract{Sequences with low auto-correlation property have been applied in code-division multiple access communication systems, radar and cryptography. Using the inverse Gray mapping, a quaternary sequence of even length $N$ can be obtained from two binary sequences of the same length, which are called  component sequences. In this paper, using interleaving method, we present several classes of component sequences from twin-prime sequences pairs or GMW sequences pairs given by Tang and Ding in 2010; two, three or four binary sequences defined by cyclotomic classes of order $4$. Hence we can obtain new classes of quaternary sequences, which are different from known ones, since known component sequences are constructed from a pair of binary sequences with optimal auto-correlation or Sidel'nikov sequences.}


\maketitle

\section{Introduction}

Binary and quaternary sequences have received a lot of attention since they are easy to be implemented as multiple-access sequences in practical communication systems, radar, and cryptography \cite{Fan96,GG2005}. For example in some communication systems, in order to acquire the desired information from the received signals, the employed sequences are required to have auto-correlation values as low as possible so as to reduce the interference and noise.
See \cite{LS03} for a good survey paper on known constructions of binary and quaternary sequences with optimal auto-correlation.

Let $s=(s(0),s(1),\cdots,s(N-1))$ and  $t=(t(0),t(1),\cdots,t(N-1))$ be two sequences of length $N$ defined over the integer residue ring $\mathbb{Z}_H=\{0,1,\cdots,H-1\}$. Then $s$ is called a binary sequence if $H=2$ or a quaternary sequence if $H=4$. The support set  of a binary sequence $s$ is defined by the set $\{0\le i<N: s(i)=1\}$.

The cross-correlation function $R_{s,t}(\tau)$ between $s$ and $t$ is defined by
\begin{eqnarray*}
R_{s,t}(\tau)&=&\sum_{i=0}^{N-1} \xi^{s(i)-t(i+\tau)},\,\,\,\,0\le\tau<N,
\end{eqnarray*}
where $\xi=\exp (2\pi\sqrt{-1}/H)$ and the subscript $i+\tau$ is performed modulo $N$. If $s=t$, $R_{s,t}(\tau)$ is called the auto-correlation function of $s$, and denoted by $R_{s}(\tau)$ for short. The
maximum out-of-phase auto-correlation magnitude of $s$ is defined as
$$R_{\max}(s)=\max\{|R_{s}(\tau)|: 1\le\tau<N\}.$$

For a quaternary sequence $s$ of odd length $N$, its maximum out-of-phase auto-correlation magnitude $R_{\max}(s)$ introduced above, is greater than or equal to 1, i.e., $R_{\max}(s)\ge 1$. Up to now, the only known class with $R_{\max}(s)=1$ was proposed in \cite{S1997}. This class of sequences has odd length $N={q+1\over 2}$ and is constructed from odd perfect sequences \cite{Luke95} of length $q+1$, where $q\equiv 1(\bmod~4)$ is an odd prime power. The next smallest values for the maximum out-of-phase auto-correlation magnitude of a quaternary sequence of odd length are as follows:
\begin{itemize}
\item $R_{\max}(s)=\sqrt{5}$ for $N\equiv 1(\bmod~4)$ \cite{Green2001,LTH2014,Sidelnikov,TL09,YK2010};
    \item $R_{\max}(s)=3$ for  $N\equiv 3(\bmod~4)$ \cite{LTH2014,YK2011}.
\end{itemize}
Those constructions were mainly based on cyclotomy or interleaving technique \cite{GG2005}.

For the case of even length $N$, a sequence $s$ is
called {\it optimal} if $R_{\max}(s)=2$ \cite{TD10}. In \cite{KJKN2009,LS2000}, optimal quaternary sequences of length $q-1$ were obtained from Sidel'nikov sequences, $q$ being an odd prime power. Using the inverse Gray mapping, a quaternary sequence of even length can be obtained from two binary sequences of the same length, which are called component sequences in this paper. Several constructions of component sequences via interleaving Legendre sequences \cite{KJKN}, or binary sequences with ideal auto-correlation \cite{JKKN2009}, were presented to design optimal quaternary sequences. By extending the constructions in \cite{KJKN} and \cite{JKKN2009}, Tang and Ding developed a generic construction of component sequences which works for any pair of ideal sequences of the same length.

The objective of this paper is to obtain new more component sequences via interleaving technique.
It will be seen later that the resultant component sequences include a pair of  non-ideal sequences, and
lead to  new optimal quaternary sequences under the inverse Gray mapping.
More precisely, our two binary component sequences can be defined by the following sequences:
\begin{itemize}
\item
 Twin-prime sequences pairs and GMW sequences pairs given by Tang and Gong in 2010 \cite{TG10};
\item Two, three or four binary sequences defined by cyclotomic classes of order $4$ with respect to the integer residue ring $\mathbb{Z}_n$, $n$ begin an odd prime.
\end{itemize}
Compared with optimal quaternary sequences given by \cite{JKKN2009,KJKN,KJKN2009,TD10}, ours have different auto-correlation functions. Examples applying non-ideal sequences to design optimal quaternary sequences are also given.

This paper is organized as follows. In Section 2, interleaving method and  Gray mapping will be briefly introduced. In Section 3, using the inverse Gray mapping to two binary sequences, a generic  construction of quaternary sequences of even length will be proposed. In Section 4, as an application of the generic construction, we first recall known constructions of component sequences, and then present some new component sequences  derived from GMW sequences pairs and twin-prime sequences pairs given in \cite{TG10} and by using two, three and four different sequences defined by  cyclotomic classes of order $4$ with respect to the integer residue ring $\mathbb{Z}_n$, $n$ being an odd prime, respectively. In Section 5, we will give three examples to illustrate our results.  Finally, some concluding remarks will be given in Section 6.

\section{Preliminaries}

\subsection{Interleaved sequences of length $2n$}

In this subsection, we briefly introduce to the representation of an interleaved sequence of length $2n$. Please refer to \cite{GG2005} for more details for the interleaving method.

Let $n$ be a positive integer. Assume that $a_i=(a_i(0),a_i(1),\cdots,a_i(n-1))$ is a sequence of length $n$, $i=0,1$, and $g=(g_0,g_1)$  is a sequence defined over $\mathbb{Z}_n$. Define a matrix $(u_{i,j})_{n\times 2}$:
\begin{eqnarray*}(u_{i,j})_{n\times 2}=
\left(
      \begin{array}{cc}
        a_0(g_0) & a_1(g_1) \\
        a_0(g_0+1) & a_1(g_1+1)  \\
        \vdots & \vdots \\
        a_0(g_0+n-1) & a_1(g_1+n-1) \\
      \end{array}
    \right).
\end{eqnarray*}

Concatenating the successive rows of the matrix above, an interleaved sequence
$u$ of length $2n$ is obtained as follows
$$u(2i+j)=u_{i,j},\,\,\,\,0\le i<N,0\le j<2.$$
For convenience, denote $u$ as
$$u=I(L^{g_0}(a_0),L^{g_1}(a_1)),$$
where $I$ is the interleaving operator, and $L^{g_i}(a_i)=(a_i(g_i),a_i(g_i+1),\cdots,a_i(g_i+n-1))$. The sequences $a_0$ and $a_1$ are called the column sequences of $u$. Let
$$v=(L^{f_0}(b_0),L^{f_1}(b_1)).$$

Consider the $\tau$ shifted version
$L^{\tau}(v)$ of $v$, where $\tau=2\tau_1+\tau_2$ $(0\le
\tau_1<n, 0\le \tau_2<2)$, we have
\begin{eqnarray*}L^{\tau}(v)&=&\left\{\begin{array}{ll}I(L^{f_0+\tau_1}(b_0)
,L^{f_1+\tau_1}(b_1)), & \tau=2\tau_1,\\
I(L^{f_1+\tau_1}(b_1),L^{f_0+\tau_1+1}(b_0)), & \tau=2\tau_1+1.\end{array}\right.\end{eqnarray*}
It then follows that the cross-correlation function between $u$ and $v$ at the shift $\tau$ is given by
\begin{eqnarray}\label{eq even correlation}
\begin{array}{rcl}
R_{u,v}(\tau)
&=&\left\{\begin{array}{ll}R_{a_0,b_0}(\tau_1+f_0-g_0)+R_{a_1,b_1}(\tau_1+f_1-g_1)
, & \tau=2\tau_1,\\
R_{a_0,b_1}(\tau_1+f_1-g_0)+R_{a_1,b_0}(\tau_1+1+f_0-g_1), & \tau=2\tau_1+1.\end{array}\right.\end{array}
\end{eqnarray}

\subsection{Gray mapping}

The well-known Gray mapping $\phi: \mathbb{Z}_4\rightarrow \mathbb{Z}_2\times
\mathbb{Z}_2$ is defined as \begin{eqnarray*}\phi(0)=(0,0), \phi(1)=(0,1),
\phi(2)=(1,1), \phi(3)=(1,0).\end{eqnarray*} Using the inverse Gray mapping
$\phi^{-1}: \mathbb{Z}_2\times \mathbb{Z}_2\rightarrow \mathbb{Z}_4$, i.e.,
\begin{eqnarray*}\phi^{-1}(0,0)=0,\phi^{-1}(0,1)=1,
  \phi^{-1}(1,1)=2,\phi^{-1}(1,0)=3,\end{eqnarray*} any quaternary sequence
$u=(u(0),u(1),\cdots,u(N-1))$ can be obtained from two binary
sequences $c=(c(0),c(1),$ $\cdots,c(N-1))$ and
$d=(d(0),d(1),\cdots,d(N-1))$ as follows:
\begin{eqnarray}\label{eq u}u(i)=\phi^{-1}(c(i),d(i)), 0\le i<N.\end{eqnarray}
Here the binary sequences $c$ and $d$ are called {\it component sequences} of $u$.

Transforming the sequence $u$ into its complex valued version,
$$\xi^{u(i)}={1\over 2}(1+\xi)(-1)^{c(i)}+{1\over 2}(1-\xi)(-1)^{d(i)}, 0\le i<N,$$
where $\xi=\sqrt{-1}$, Krone and Sarwate \cite{KS1984} observed the following
result.

\begin{lemma}[\cite{KS1984}]\label{le quaternary}
The auto-correlation function of $u$ is given by
\begin{eqnarray}\label{eq sarwate}R_{u}(\tau)={1\over 2}(R_{c}(\tau)+R_{d}(\tau))+
{\xi\over 2}(R_{c,d}(\tau)-R_{d,c}(\tau)),~~~~0\le\tau<N.\end{eqnarray}
\end{lemma}

\section{Generic construction of quaternary sequences }

In this section, we present a procedure for the construction of
quaternary sequences with optimal auto-correlation.

{\bf Construction I: Construction of quaternary sequence via Gray mapping.}

\begin{enumerate}
\item [1)] Let $n$ be an odd integer, $N=2n$, and $\lambda={n+1\over 2}$. Generate four binary sequences $a_i$ of length $n$, $0\le i\le 3$, and a binary sequence $e=(e(0),e(1),e(2))$, $e(j)=0,1$.
\item [2)] Define two binary sequences of length $N$:
\begin{eqnarray}\label{eq cd}c=I(a_0,e(0)+L^{\lambda}(a_1)),~~d=I(e(1)+a_2,e(2)+L^{\lambda}(a_3)),\end{eqnarray}
where $L^{\lambda}(b)+1$ denotes the complement of the sequence $L^{\lambda}(b)=(b(\lambda),b(\lambda+1),\cdots,b(\lambda+n-1))$, i.e., $L^{\lambda}(b)+1=(b(\lambda)+1,b(\lambda+1)+1,\cdots,b(\lambda+n-1)+1)$.
\item [3)] Applying the inverse Gray mapping $\phi^{-1}$ to $c$ and $d$, obtain a quaternary sequence $u$ of length $N$, where $u(i)=\phi^{-1}(c(i),d(i))$.
\end{enumerate}

We have the following result.

\begin{theorem}\label{thm main}
The auto-correlation of $u$ generated by Construction I is given by
\begin{eqnarray*}
R_u(\tau)&=&[R_{a_0}(\tau_0)+R_{a_1}(\tau_0)+R_{a_2}(\tau_0)+R_{a_3}(\tau_0)]/2\\&&+
(-1)^{e(1)}[R_{a_0,a_2}(\tau_0)-R_{a_2,a_0}(\tau_0)+(-1)^{e(0)+e(1)+e(2)}(R_{a_1,a_3}(\tau_0)-R_{a_3,a_1}(\tau_0))]\xi/2,
\end{eqnarray*}
if $\tau=2\tau_0$, and
\begin{eqnarray*}
R_u(\tau)&=&(-1)^{e(0)}[
R_{a_0,a_1}(\tau_2)+R_{a_1,a_0}(\tau_2)+
(-1)^{e(0)+e(1)+e(2)}(R_{a_2,a_3}(\tau_2)+R_{a_3,a_2}(\tau_2))]/2\\&&
+(-1)^{e(2)}[R_{a_0,a_3}(\tau_2)-R_{a_3,a_0}(\tau_2)
+(-1)^{e(0)+e(1)+e(2)}(R_{a_1,a_2}(\tau_2)-R_{a_2,a_1}(\tau_2))]\xi/2,
\end{eqnarray*}
if $\tau=2\tau_0+1$, where $\tau_2=\tau_0+\lambda$.
\end{theorem}

\begin{proof}
Calculate the auto-correlation and cross-correlation functions of $c$ and $d$. Writing $\tau=2\tau_0+\tau_1$, where $0\le\tau_0<n$ and $\tau_1=0,1$, we consider the auto-correlation of $c$ in two cases according to $\tau_1=0$ and $\tau_1=1$.

Case 1: $\tau_1=0$, by (\ref{eq even correlation}), in this case we have
\begin{eqnarray*}
R_c(\tau)&=&R_{a_0}(\tau_0)+R_{a_1}(\tau_0).
\end{eqnarray*}

Case 2: $\tau_1=1$, by (\ref{eq even correlation}) again, we have

\begin{eqnarray*}
R_c(\tau)&=&(-1)^{e(0)}R_{a_0,a_1}(\tau_0+\lambda)+(-1)^{e(0)}R_{a_1,a_0}(\tau_0+1-\lambda),\\
&=&(-1)^{e(0)}(R_{a_0,a_1}(\tau_0+\lambda)+R_{a_1,a_0}(\tau_0+\lambda)),
\end{eqnarray*}
where the second identity was due to $(\tau_0+1-\lambda)\equiv (\tau_0+\lambda)(\bmod~n)$. The following correlation functions can be similarly proved.
\begin{eqnarray*}
R_d(\tau)&=&\left\{\begin{array}{ll}R_{a_2}(\tau_0)+R_{a_3}(\tau_0), & \tau=2\tau_0,\\
(-1)^{e(1)+e(2)}(R_{a_2,a_3}(\tau_0+\lambda)+R_{a_3,a_2}(\tau_0+\lambda)), & \tau=2\tau_0+1,\end{array}\right.\\
R_{c,d}(\tau)&=&\left\{\begin{array}{ll}(-1)^{e(1)}R_{a_0,a_2}(\tau_0)+(-1)^{e(0)+e(2)}R_{a_1,a_3}(\tau_0), & \tau=2\tau_0,\\
(-1)^{e(2)}R_{a_0,a_3}(\tau_0+\lambda)+(-1)^{e(0)+e(1)}R_{a_1,a_2}(\tau_0+\lambda), & \tau=2\tau_0+1,\end{array}\right.\\
R_{d,c}(\tau)&=&\left\{\begin{array}{ll}(-1)^{e(1)}R_{a_2,a_0}(\tau_0)+(-1)^{e(0)+e(2)}R_{a_3,a_1}(\tau_0), & \tau=2\tau_0,\\
(-1)^{e(2)}R_{a_3,a_0}(\tau_0+\lambda)+(-1)^{e(0)+e(1)}R_{a_2,a_1}(\tau_0+\lambda), & \tau=2\tau_0+1.\end{array}\right.
\end{eqnarray*}
The conclusion then follows from  (\ref{eq sarwate}) and the discussion above.
\end{proof}

\begin{corollary}\label{coro main}
Let $a_0,a_1,a_2,a_3$ be four binary sequences of odd length $n$ and $e=(e(0),e(1),e(2))$ be a binary sequence.
Then $R_{\max}(u)=2$, if
\begin{eqnarray}\label{eq 1}\left\{\begin{array}{ll}
R_{a_0}(\tau_0)+R_{a_1}(\tau_0)+R_{a_2}(\tau_0)+R_{a_3}(\tau_0)\in \{0,\pm 4\},~~~~1\le\tau_0<n,\\
R_{a_0,a_2}(\tau_0)-R_{a_2,a_0}(\tau_0)+(-1)^{e(0)+e(1)+e(2)}(R_{a_1,a_3}(\tau_0)-R_{a_3,a_1}(\tau_0))=0,~~~~1\le\tau_0<n,\\
R_{a_0,a_1}(\tau_0)+R_{a_1,a_0}(\tau_0)
+(-1)^{e(0)+e(1)+e(2)}(R_{a_2,a_3}(\tau_0)+R_{a_3,a_2}(\tau_0))\in \{0,\pm 4\},~~~~0\le\tau_0<n,\\
R_{a_0,a_3}(\tau_0)+R_{a_3,a_0}(\tau_0)
+(-1)^{e(0)+e(1)+e(2)}(R_{a_1,a_2}(\tau_0)-R_{a_2,a_1}(\tau_0))=0,~~~~0\le\tau_0<n.
\end{array}\right.\end{eqnarray}

\end{corollary}

\begin{proof}
If (\ref{eq 1}) holds, then by Theorem \ref{thm main},
\begin{eqnarray*}
R_u(\tau)&=&[R_{a_0}(\tau_0)+R_{a_1}(\tau_0)+R_{a_2}(\tau_0)+R_{a_3}(\tau_0)]/2\in \{0, \pm 2\},
\end{eqnarray*}
if $\tau=2\tau_0$, and
\begin{eqnarray*}
R_u(\tau)&=&(-1)^{e(0)}[
R_{a_0,a_1}(\tau')+R_{a_1,a_0}(\tau')+
(-1)^{e(0)+e(1)+e(2)}(R_{a_2,a_3}(\tau')+R_{a_3,a_2}(\tau'))]/2\in \{0, \pm 2\},
\end{eqnarray*}
if $\tau=2\tau_0+1$, where $\tau'=\tau_0+\lambda$. Hence $R_u(\tau)\in \{0,\pm 2\}$ for all $1\le\tau<N$, i.e., $u$ is optimal.
\end{proof}

\section{Quaternary sequences from the generic construction}
In this section, we will show that our generic construction
includes some known constructions of optimal quaternary sequences as special cases, and can produce new quaternary sequences with optimal
auto-correlation. Throughout this section, suppose that $u$ is the quaternary sequence generated by
Construction I.

\subsection{Known constructions of $a_0,a_1,a_2,a_3$}

\begin{theorem}[\cite{JKKN2009}]\label{thm ideal}
Let $a_0=a_1=a_2=a_3$, which are the same ideal sequences of length $n=2^m-1$, and $e=(0,0,1)$. Then $u$ is an optimal quaternary sequence, and for $1\le\tau<2n$,
\begin{eqnarray*}
R_u(\tau)&=&\left\{\begin{array}{ll}-2, & \tau=2\tau_0,\\
0, & \tau=2\tau_0+1.\end{array}\right.
\end{eqnarray*}
\end{theorem}

Theorem \ref{thm ideal} was generalized by Tang and Ding as follows.

\begin{theorem}[\cite{TD10}] Let $a_0=a_1$ and $a_2=a_3$ be ideal sequences of the same length $n$, i.e., $R_{a_0}(\tau_0)=R_{a_2}(\tau_0)=-1$, $1\le\tau_0<n$. Let $e=(0,0,1)$. Then $u$  is an optimal quaternary sequence with auto-correlation function
\begin{eqnarray*}
R_u(\tau)&=&\left\{\begin{array}{ll}-2, & \tau=2\tau_0,\\
0, & \tau=2\tau_0+1.\end{array}\right.
\end{eqnarray*}
\end{theorem}

In \cite{KJKN} and \cite{TD10}, the following result has been obtained by choosing the Legendre sequences pair (Please refer to \cite{TG10} for more details).

\begin{theorem}
[\cite{KJKN,TD10}] Let $s$ and $t$ be the Legendre sequences pair of odd prime length $n$. Let $e=(0,0,1)$ and
\begin{eqnarray*}(a_0,a_1,a_2,a_3)\in\{
(s,t,s,t),(s,t,t,s),(t,s,t,s),(t,s,s,t)\}.
\end{eqnarray*}
Then $u$  is an optimal quaternary sequence with auto-correlation function
\begin{eqnarray*}
R_u(\tau)&=&\left\{\begin{array}{ll}-2, & \tau=2\tau_0,\\
0, & \tau=2\tau_0+1.\end{array}\right.
\end{eqnarray*}
\end{theorem}

\begin{remark}
From known constructions above, $a_0,a_1,a_2,a_3$ were defined by one or two binary sequences with optimal auto-correlation. In the next subsections, we will present new constructions of $a_0,a_1,a_2,a_3$, some of which have non-optimal auto-correlation functions. Those new $a_0,a_1,a_2,a_3$ satisfy (\ref{eq 1}), and can be used to obtain optimal quaternary sequences $u$.
\end{remark}

\subsection{New constructions of $a_0,a_1,a_2,a_3$ using a sequence pair}

Using the twin-prime sequences pairs and the GMW-sequences pairs given in \cite{TG10}, the following results can be obtained from Corollary \ref{coro main}.

\begin{theorem} Let $s$ and $t$ be the twin-prime sequences pair of length $p(p+2)$. Let $e=(e(0),e(1),e(2))$ satisfy $e(0)+e(1)+e(2)\equiv 1(\bmod~2)$ and
\begin{eqnarray*}(a_0,a_1,a_2,a_3)\in\{
(s,t,s,t),(s,t,t,s),(t,s,t,s),(t,s,s,t)\}.
\end{eqnarray*}
Then $u$ given by Construction I is an optimal quaternary sequence with auto-correlation function
\begin{eqnarray*}
R_u(\tau)&=&\left\{\begin{array}{ll}-2, & \tau=2\tau_0, \tau_0\equiv 0(\bmod~p+2),\\
2, & \tau=2\tau_0, \tau_0\not\equiv 0(\bmod~p+2),\\
0, & \tau=2\tau_0+1, \end{array}\right.
\end{eqnarray*}
\end{theorem}

\begin{theorem}\label{thm six} Let $s$ and $t$ be the GMW sequences pair of length $2^{2k}-1$. Let $e=(e(0),e(1),e(2))$ satisfy $e(0)+e(1)+e(2)\equiv 1(\bmod~2)$ and
\begin{eqnarray*}(a_0,a_1,a_2,a_3)\in\{
(s,t,s,t),(s,t,t,s),(t,s,t,s),(t,s,s,t)\}.
\end{eqnarray*}
Then $u$ given by Construction I is an optimal quaternary sequence with auto-correlation function
\begin{eqnarray*}
R_u(\tau)&=&\left\{\begin{array}{ll}-2, & \tau=2\tau_0, \tau_0\equiv 0(\bmod~2^k+1),\\
2, & \tau=2\tau_0, \tau_0\not\equiv 0(\bmod~2^k+1),\\
0, & \tau=2\tau_0+1.\end{array}\right.
\end{eqnarray*}
\end{theorem}

\begin{remark}\label{re 2}
By choosing the twin-prime sequences pairs and GMW sequences pairs, the quaternary sequence $u$ given by Construction I are different from the quaternary sequence given by Theorem 6 of \cite{TD10}, since the auto-correlation function of our sequence take values $0,\pm 2$, and that of the sequence in \cite{TD10} takes values $0,-2$.
\end{remark}

\subsection{Constructions of $a_0,a_1,a_2,a_3$ using cyclotomic classes of order $4$}

Assume that $n=4f+1=x^2+4y^2$ is an odd prime, where $f$, $x$ and $y$ are integers. Let $D_0,D_1,D_2,D_3$ be  the cyclotomic classes of order $4$ with respect to $\mathbb{Z}_n$ (See Appendix A). Let $s_1,s_2,s_3,s_4,s_5,s_6$ be six binary sequences of length $n$ with support sets $D_0\cup D_1$, $D_0\cup D_2$, $D_0\cup D_3$, $D_1\cup D_2$, $D_1\cup D_3$, $D_2\cup D_3$, respectively.

In this subsection, we will present new constructions of $a_0,a_1,a_2,a_3$ choosing from $s_1,s_2,s_3,s_4,s_5,s_6$, whose auto-correlation and cross-correlation functions are given in Appendix A.  The following discussion are divided into two cases: $f$ odd and $f$ even.

\begin{theorem} Let $f$ be odd, and $y=-1$. Let $e=(e(0),e(1),e(2))$ satisfy $e(0)+e(1)+e(2)\equiv 0(\bmod~2)$ and
    \begin{eqnarray*}(a_0,a_1,a_2,a_3)\in
\left\{\begin{array}{l}
(s_2,s_1,s_2,s_1),(s_1,s_2,s_1,s_2),
(s_6,s_2,s_6,s_2),(s_2,s_6,s_2,s_6),\\  (s_5,s_4,s_5,s_4),(s_4,s_5,s_4,s_5),
(s_3,s_5,s_3,s_5),(s_5,s_3,s_5,s_3)\\
\end{array}\right\}.
\end{eqnarray*}
Then $u$ given by Construction I is an optimal quaternary sequence, i.e., for $1\le\tau<2n$, $R_u(\tau)=\pm 2$.
\end{theorem}

\begin{proof}Note that $a_0=a_2$ and $a_1=a_3$, where $(a_0,a_1)\in
\{(s_2,s_1),(s_1,s_2), (s_6,s_2),(s_2,s_6),(s_5,s_4),(s_4,s_5)$, $(s_3,s_5),(s_5,s_3)\}$. By Theorem \ref{thm main}, the auto-correlation function of $u$ is reduced as
\begin{eqnarray*}
R_u(\tau)&=&\left\{\begin{array}{ll}R_{a_0}(\tau_0)+R_{a_1}(\tau_0), & \tau=2\tau_0,\\
(-1)^{e(0)}[R_{a_0,a_1}(\tau_0+\lambda)+R_{a_1,a_0}(\tau_0+\lambda)], & \tau=2\tau_0+1.\end{array}\right.
\end{eqnarray*}
Using the values of auto-correlation and cross-correlation functions of $a_0$ and $a_1$ obtained in Lemma \ref{le reduced} and Theorem \ref{thm odd} in Appendix A, the result follows immediately.
\end{proof}

\begin{theorem} Let $f$ be odd, and $y=-1$. Let $e=(e(0),e(1),e(2))$ with $e(0)+e(1)+e(2)\equiv 0(\bmod~2)$ and
    \begin{eqnarray*}(a_0,a_1,a_2,a_3)\in
\left\{\begin{array}{l}
(s_1,s_2,s_2,s_1),(s_2,s_1,s_1,s_2),
(s_2,s_6,s_6,s_2),(s_6,s_2,s_2,s_6), \\
 (s_4,s_5,s_5,s_4),(s_5,s_4,s_4,s_5),(s_5,s_3,s_3,s_5),(s_3,s_5,s_5,s_3)
\end{array}\right\}.
\end{eqnarray*}
Then $u$ given by Construction I is an optimal quaternary sequence, i.e., for $1\le\tau<2n$, $R_u(\tau)=\pm 2$.
\end{theorem}

\begin{proof}Note that $a_0=a_3$ and $a_1=a_2$,  where $(a_0,a_1)\in
\{(s_2,s_1),(s_1,s_2),
(s_6,s_2),(s_2,s_6),(s_5,s_4),(s_4,s_5)$, $(s_3,s_5),(s_5,s_3)\}$. By Theorem \ref{thm main}, the auto-correlation function of $u$ is reduced as
\begin{eqnarray*}
R_u(\tau)&=&\left\{\begin{array}{ll}R_{a_0}(\tau_0)+R_{a_1}(\tau_0), & \tau=2\tau_0,\\
(-1)^{e(0)}[R_{a_0,a_1}(\tau_0+\lambda)+R_{a_1,a_0}(\tau_0+\lambda)], & \tau=2\tau_0+1.\end{array}\right.
\end{eqnarray*}
Based on the auto-correlation and cross-correlation functions of $a_0$ and $a_1$ obtained in Lemma \ref{le reduced} and Theorem \ref{thm odd} in Appendix A, the result follows immediately.
\end{proof}

\begin{theorem}\label{coro} Let $f$ be odd and $y=-1$.  Let $e=(e(0),e(1),e(2))$ with $e(0)+e(1)+e(2)\equiv 1(\bmod~2)$ and
    \begin{eqnarray*}(a_0,a_1,a_2,a_3)\in\left\{\begin{array}{l}(s_2,s_1,s_6,s_2),(s_2,s_6,s_1,s_2),
(s_5,s_3,s_4,s_5),(s_5,s_4,s_3,s_5),\\
(s_6,s_2,s_2,s_1),(s_1,s_2,s_2,s_6),
(s_3,s_5,s_5,s_4),(s_4,s_5,s_5,s_3)\end{array}\right\}.
\end{eqnarray*}
Then $u$ given by Construction I is an optimal quaternary sequence, i.e., for $1\le\tau<2n$, $R_u(\tau)=\pm 2$.
\end{theorem}

\begin{proof}
By Theorem \ref{thm main}, the result follows immediately by using the auto-correlation and cross-correlation functions of $s_1,s_3,s_4$ and $s_6$ given in Lemma \ref{le reduced} and Theorem \ref{thm odd} in Appendix A.
\end{proof}

\begin{theorem}\label{thm even_f} Let $f$ be even and $x=\pm 1$.  Let $e=(e(0),e(1),e(2))$ with $e(0)+e(1)+e(2)\equiv 0(\bmod~2)$ and
\begin{eqnarray*}(a_0,a_1,a_2,a_3)\in \left\{\begin{array}{l}(s_6,s_3,s_4,s_1),(s_6,s_4,s_3,s_1),(s_4,s_6,s_3,s_1),(s_3,s_6,s_4,s_1),\\
(s_4,s_1,s_6,s_3),(s_6,s_4,s_1,s_3),(s_1,s_4,s_6,s_3),(s_4,s_6,s_1,s_3),\\
(s_3,s_1,s_6,s_4),(s_6,s_3,s_1,s_4),(s_1,s_3,s_6,s_4),(s_3,s_6,s_1,s_4),\\
(s_4,s_1,s_3,s_6),(s_3,s_1,s_4,s_6),(s_1,s_3,s_4,s_6),(s_1,s_4,s_3,s_6)\end{array}\right\}.
\end{eqnarray*}
Then $u$ given by Construction I is an optimal quaternary sequence, i.e., for $1\le\tau<2n$, $R_u(\tau)=\pm 2$.
\end{theorem}

\begin{proof} Note that for any $i\ne j$, $R_{s_i,s_j}(\tau)=R_{s_j,s_i}(\tau)$ holds for all $0\le\tau<n$ (See Lemma \ref{le reduced} in Appendix A). That is to say, $0\le i\ne j\le 3$, $R_{a_i,a_j}(\tau)=R_{a_j,a_i}(\tau)$ holds for all $0\le\tau<n$. Hence by Theorem \ref{thm main}, the auto-correlation function of $u$ is given by
\begin{eqnarray*}
R_u(\tau)&=&\left\{\begin{array}{ll}[R_{a_0}(\tau_0)+R_{a_1}(\tau_0)+R_{a_2}(\tau_0)+R_{a_3}(\tau_0)]/2, & \tau=2\tau_0,\\
(-1)^{e(0)}[
R_{a_0,a_1}(\tau_0+\lambda)+R_{a_2,a_3}(\tau_0+\lambda)], & \tau=2\tau_0+1.
\end{array}\right.
\end{eqnarray*}
The result follows immediately from the auto-correlation and cross-correlation functions of $s_1,s_3,s_4$ and $s_6$ given by Lemma \ref{le reduced} and Theorem \ref{thm even} in Appendix A.
\end{proof}

\section{Examples}

In this section, we will give three examples of our new constructions of quaternary sequences with optimal auto-correlation.

\begin{example}
Define two binary sequences of length $25$ as follows:
\begin{eqnarray*}
a_0=a_2=(0, 0, 0, 1, 1, 1, 0, 1, 0, 1, 1, 0, 1, 0, 1, 1, 0, 1, 0, 1, 1, 1, 0, 0, 0),\\
a_1=a_3=(1, 0, 0, 1, 1, 0, 1, 0, 0, 0, 0, 0, 1, 0, 0, 0, 0, 0, 1, 0, 1, 1, 0, 0, 1).
\end{eqnarray*}

With the help of Magma Program, the auto-correlation and cross-correlation functions of $a_0$ and $a_1$ are given by
\begin{eqnarray*}
(R_{a_0}(\tau))_{\tau=0}^{24}&=&(  25, -3, 5, -3, -7, 5, -7, 1, -3, 1, 1, 1, -3, -3, 1, 1, 1, -3, 1, -7, 5, -7, -3, 5, -3),\\
(R_{a_1}(\tau))_{\tau=0}^{24}&=& (25, 1, -3, 5, 5, -3, 5, 1, 5, 1, -3, -3, 1, 1, -3, -3, 1, 5, 1, 5, -3, 5, 5, -3, 1 ),\\
(R_{a_0,a_1}(\tau))_{\tau=0}^{24}&=& (1, 1, -7, 1, -3, -3, -3, -3, 5, 5, 1, 5, -3, -3, 5, 1, 5, 5, -3, -3, -3, -3, 1, -7, 1 ),\\
(R_{a_1,a_0}(\tau))_{\tau=0}^{24}&=& ( 1, 1, -7, 1, -3, -3, -3, -3, 5, 5, 1, 5, -3, -3, 5, 1, 5, 5, -3, -3, -3, -3, 1, -7, 1).
    \end{eqnarray*}
It can be seen that both $a_0$ and $a_1$ are  non-ideal sequences. By Theorem \ref{thm main}, we can obtain a quaternary sequence $u$
\begin{eqnarray*}u&=&(0,1,0,1,0,1,2,
1,2,1,2,3,0,1,2,3,0,3,2,1,2,1,0,3,2,3,0,1,2,\\&&
1,2,3,0,3,2,1,0,3,2,1,2,1,2,1,0,1,0,1,0,3).
\end{eqnarray*}
with auto-correlation
\begin{eqnarray*}
    (R_{u}(\tau))_{\tau=1}^{49}&=&(0,-2,0,2,0,2,0,-2,0,2,0,-2,0,2,0,2,0,2,0,-2,0,\\&&
    -2,0,-2,0,-2,0,-2,0,-2,0,2,0,2,0,2,0,-2,0,2,0,-2,0,2,0,2,0,-2,0).
\end{eqnarray*}
\end{example}

\begin{example}
Let $\alpha$ be a primitive element of the finite field $F_{2^6}$ generated by
the primitive polynomial $f(x)=x^6+x+1$ and $f(\alpha)=0$. Let
$a_0=(a_0(0),a_0(1),\cdots,a_0(62))$ be the m-sequence of length $63$, where
$a_0(i)=Tr_1^6(\alpha^i)=\alpha^i+\alpha^{2i}+\alpha^{4i}+\alpha^{8i}+\alpha^{16i}+\alpha^{32i}$, i.e.,
\begin{eqnarray*}a_0=a_2&=&(0, 0, 0, 0, 0, 1, 0, 0, 0, 0, 1, 1, 0, 0, 0, 1, 0, 1, 0, 0,1, 1, 1, 1,0, 1, 0, 0,
0,1, 1, 1, 0, 0, 1, \\&& 0, 0, 1, 0, 1, 1, 0, 1, 1, 1, 0, 1, 1, 0, 0,  1, 1, 0, 1, 0, 1, 0, 1, 1, 1, 1, 1, 1),\end{eqnarray*} and its modification $a_1$ given in \cite{TG10} equal to
\begin{eqnarray*}a_1=a_3&=&(1, 0, 0, 0, 0, 1, 0, 0, 0, 1, 1, 1, 0, 0, 0, 1, 0, 1, 1, 0, 1, 1, 1,1,0, 1, 0,1,0, 1, 1, 1, 0, 0, 1,\\&& 0, 1, 1, 0, 1,1, 0, 1, 1, 1, 1, 1, 1, 0, 0,1, 1, 0, 1, 1, 1, 0, 1,1, 1, 1, 1, 1).\end{eqnarray*} Then by Theorem \ref{thm six}, one has
\begin{eqnarray*}
u&=&(0, 1, 0, 1, 0, 3, 0, 1, 0, 3, 2, 3, 0, 1, 0,3, 0, 3, 0, 1, 2,3, 2,3, 0, 3, 0, 3, 0, 3,2, 3, 0,1,2,
\\&&1, 0, 3,0,3, 2, 1,2, 3, 2, 3, 2, 3, 0,
 1, 2, 3, 0, 3, 0, 3, 0, 3,2,3, 2, 3, 2, 3, 0, 1, 0, 1, 2,1, \\&& 0, 1, 0, 3, 2,
 1, 0,1,2, 1,   2, 3, 0, 3, 2,3, 2, 1,2, 1, 0, 1, 2, 3, 2, 1, 0, 3, 0,3, 2, 1, 2, 3, 0, \\&&
3, 2, 3, 0, 3, 2, 1,0, 3, 2, 1, 2, 3, 2, 1, 2, 3, 2, 3, 2, 3).
\end{eqnarray*}
Hence we have
\begin{eqnarray*}
(R_{u}(\tau))_{\tau=1}^{125}&=&
(0, 2, 0, 2, 0, 2, 0, 2, 0, 2, 0, 2, 0,  2, 0, 2, 0, -2, 0, 2,0, 2, 0, 2,0, 2, 0,\\&&
2,0, 2, 0, 2, 0, -2, 0, 2, 0, 2, 0, 2,
 0, 2,0, 2,0, 2, 0, 2, 0,2, 0, 2, 0, -2, 0, 2, \\&& 0, 2, 0, 2, 0, 2, 0,2,0, 2, 0, 2,0, 2,0,-2, 0, 2, 0,2,0,2, 0, 2, 0, 2,0,2,
2, 0,\\&&2, 0, 0, -2, 0,2, 0, 2, 0, 2, 0, 2,0, 2,0, 2, 0, 2, 0, 2,0, -2, 0, 2,0, 2, 0,2, 0, 2, \\&&
0, 2, 0, 2, 0, 2, 0, 2, 0),\end{eqnarray*}
i.e., the out-of-phase auto-correlation of $u$ takes values $0,\pm 2$.
\end{example}

\begin{example}
Let $\alpha=3$ be a generator of the multiplicative group of the integer residue ring $\mathbb{Z}_{17}$. Then the cyclotomic classes of order $4$ with respect to $\mathbb{Z}_{17}$ are given as follows:
\begin{eqnarray*}
C_0&=&\{ 1, 4, 13, 16 \},\\
C_1&=&\{ 3, 5, 12, 14 \},\\
C_2&=&\{ 2, 8, 9, 15 \},\\
C_3&=&\{ 6, 7, 10, 11 \}.
\end{eqnarray*}
It is easy to check that the following sequences
\begin{eqnarray*}
s_1&=&( 0, 1,  0, 1, 1, 1,  0,  0,  0,  0,  0,  0, 1, 1, 1,  0, 1),\\
s_3&=&( 0, 1,  0,  0, 1,  0, 1, 1,  0,  0, 1, 1,  0, 1,  0,  0, 1),\\
s_4&=&( 0,  0, 1, 1,  0, 1,  0,  0, 1, 1,  0,  0, 1,  0, 1, 1,  0),\\
s_6&=&( 0,  0, 1,  0,  0,  0, 1, 1, 1, 1, 1, 1,  0,  0,  0, 1,  0).
\end{eqnarray*}
with support sets  $C_0\cup C_1$, $C_0\cup C_3$, $C_1\cup C_2$ and $C_2\cup C_3$ respectively are non-optimal binary sequences of length $17$.
Take $a_0=s_6$, $a_1=s_3$, $a_2=s_4$, $a_4=s_1$, and $e=(0,0,0)$. By Theorem \ref{thm even_f}, the quaternary sequence $u$ is equal to
\begin{eqnarray*}
u=(0, 0, 0, 3, 2, 3, 1, 1, 0, 2, 1, 1, 3, 0, 3, 2, 2, 0, 2, 2, 3, 0, 3, 1, 1, 2, 0, 1, 1, 3, 2, 3, 0, 0),
\end{eqnarray*}
which has the out-of-phase auto-correlation function:
\begin{eqnarray*}(R_u(\tau))_{\tau=1}^{33}&=&(2,-2,-2,-2,-2,-2,-2,-2,2,-2,-2,-2,2,-2,2,\\
&&~-2,2,-2,2,-2,2,-2,-2,-2,2,-2,-2,-2,-2,-2,-2,-2,2).\end{eqnarray*}
\end{example}

\section{Conclusion}

Using the inverse Gray mapping and interleaving method, the authors in \cite{TD10}
proposed a construction of multiple-access quaternary sequences of even length with optimal magnitude by choosing arbitrary two ideal sequences of the same length, which is a generalization of \cite{JKKN2009,KJKN}. While in this
paper, we constructed component sequences via interleaving:
Twin-prime sequences pairs and GMW sequences pairs given by Tang and Gong in 2010; or two, three or four binary sequences defined by cyclotomic classes of order $4$. Compared with those sequences given in \cite{TD10}, our proposed sequences can be defined by using non-ideal binary sequences and have different auto-correlation functions.

\Acknowledgements{The work of Wei Su was supported by the National Science Foundation of China under Grant (No 61402377), and in part supported by the Open Research Subject of Key Laboratory (Research Base) of Digital Space Security szjj2014-075, and Science and Technology on Communication Security Laboratory Grant 9140C110302150C11004. The work of Yang Yang was supported by the National Science Foundation of China under Grants (Nos. 61401376 and 11571285), and the Application Fundamental Research Plan Project of Sichuan Province under Grant 2016JY0160. The work of Zhengchun Zhou and Xiaohu Tang was supported by the National Science Foundation of China under Grants (Nos. 61672028 and 61325005).
}

\begin{appendix}
\section{The auto-correlation and cross-correlation of $s_i, 1\le i\le 6$}

In this section, we will first review the cyclotomic classes of order $4$ and then discuss the auto-correlation and cross-correlation of $s_i$ defined by cyclotomic classes.

Assume that $n=4f+1=x^2+4y^2$ is a prime, where $f$, $x$ and $y$ are integers.  Let $\alpha$ be a generator of the multiplicative group of the integer residue ring $\mathbb{Z}_n$, and let $C_i=\{\alpha^{4j+i}: 0\le j<f \}$, $0\le i<4$. Those $C_i, 0\le i<4$, are called the cyclotomic classes of order $4$ with respect to $\mathbb{Z}_n$. The cyclotomic numbers of order $4$, denoted $(i, j)$, are defined as
$$(i,j)=|(C_i+1)\cap C_j|.$$

The cyclotomic numbers of order $4$ are given in \cite{Stor}.

\begin{lemma}[\cite{Stor}]\label{le fourorder}
\begin{itemize}
\item For odd $f$, the sixteen cyclotomic numbers are given by the following table, where $A=\frac{n-7+2x}{16}$, $B=\frac{n+1+2x-8y}{16}$, $C=\frac{n+1-6x}{16}$,
$D=\frac{n+1+2x+8y}{16}$, $E=\frac{n-3-2x}{ 16}$.

\begin{table}[ht]
\begin{center}
\caption{$f$ odd }
\begin{tabular}{|c|c|c|c|c|}
\hline $(i,j)$ & $0$ & $1$ & $2$ & $3$  \\   \hline $0$ & $A$ &
$B$ & $C$ & $D$ \\   \hline $1$ & $E$ & $E$ & $D$ & $B$ \\   \hline
$2$
& $A$ & $E$ & $A$ & $E$ \\   \hline $3$ & $E$ & $D$ & $B$ & $E$ \\
\hline
\end{tabular}
\end{center}
\end{table}

 \item For even $f$, the sixteen cyclotomic numbers are given by the following table, where $A=\frac{n-11-6x}{16}$, $B=\frac{n-3+2x+8y}{16}$, $C=\frac{n-3+2x}{ 16}$, $D=\frac{n-3+2x-8y}{16}$, $E=\frac{n+1-2x}{ 16}$.

\begin{table}[ht]
\begin{center}
\caption{$f$ even}
\begin{tabular}{|c|c|c|c|c|}
\hline $(i,j)$ & $0$ & $1$ & $2$ & $3$  \\   \hline $0$ & $A$ &
$B$ & $C$ & $D$ \\  \hline $1$ & $B$ & $D$ & $E$ & $E$ \\  \hline
$2$ & $C$ & $E$ & $C$ & $E$ \\  \hline $3$ & $D$ & $E$ & $E$ & $B$ \\
\hline
\end{tabular}
\end{center}
\end{table}
\end{itemize}
\end{lemma}

Let $s_1,s_2,s_3,s_4,s_5,s_6$ be six binary sequences of odd prime length $n$ with support sets $D_0\cup D_1$, $D_0\cup D_2$, $D_0\cup D_3$, $D_1\cup D_2$, $D_1\cup D_3$, $D_2\cup D_3$, respectively. The auto-correlation and cross-correlation of $s_i, 1\le i\le 6$, are listed in Tables \ref{tab odd f} and \ref{tab even f} respectively.

Let $i_0i_1i_2i_3$ and $j_0j_1j_2j_3$ be two permutations of $0, 1, 2, 3$.
Let $s_i$ and $s_j$ be two binary sequences with support sets $D_{i_0}\cup D_{i_1}$ and $D_{j_0}\cup D_{j_1}$, respectively. The cross-correlation of $s_i$ and $s_j$ at shift $\tau\in D_k$ is equal to
\begin{eqnarray}\label{eqn_cross_st}
\begin{array}{rcl}
R_{s_i,s_j}(\tau)&=&(-1)^{s_i(0)+s_j(\tau)}+(-1)^{s_i(n-\tau)+s_j(0)}+\Delta_{s_i,s_j}(\tau)\\
&=&(-1)^{s_j(\tau)}+(-1)^{s_i(n-\tau)}+\Delta_{s_i,s_j}(\tau).
\end{array}\end{eqnarray}
where
\begin{eqnarray}\label{eqn_auto_st}
\begin{array}{rcl}
\Delta_{s_i,s_j}(\tau)&=&~~|\{1\le i<n: i\in D_{i_0}, i+\tau\in D_{j_0}\}|+|\{1\le i<n: i\in D_{i_0}, i+\tau\in D_{j_1}\}|\\
&&-|\{1\le i<n: i\in D_{i_0}, i+\tau\in D_{j_2}\}|-|\{1\le i<n: i\in D_{i_0}, i+\tau\in D_{j_3}\}|\\
&&+|\{1\le i<n: i\in D_{i_1}, i+\tau\in D_{j_0}\}|+|\{1\le i<n: i\in D_{i_1}, i+\tau\in D_{j_1}\}|\\
&&-|\{1\le i<n: i\in D_{i_1}, i+\tau\in D_{j_2}\}|-|\{1\le i<n: i\in D_{i_1}, i+\tau\in D_{j_3}\}|\\
&&-|\{1\le i<n: i\in D_{i_2}, i+\tau\in D_{j_0}\}|-|\{1\le i<n: i\in D_{i_2}, i+\tau\in D_{j_1}\}|\\
&&+|\{1\le i<n: i\in D_{i_2}, i+\tau\in D_{j_2}\}|+|\{1\le i<n: i\in D_{i_2}, i+\tau\in D_{j_3}\}|\\
&&-|\{1\le i<n: i\in D_{i_3}, i+\tau\in D_{j_0}\}|-|\{1\le i<n: i\in D_{i_3}, i+\tau\in D_{j_1}\}|\\
&&+|\{1\le i<n: i\in D_{i_3}, i+\tau\in D_{j_2}\}|+|\{1\le i<n: i\in D_{i_3}, i+\tau\in D_{j_3}\}|\\
&=&~~(j_0-k,i_0-k)+(j_1-k,i_0-k)-(j_2-k,i_0-k)-(j_3-k,i_0-k)\\
&&+(j_0-k,i_1-k)+(j_1-k,i_1-k)-(j_2-k,i_1-k)-(j_3-k,i_1-k)\\
&&-(j_0-k,i_2-k)-(j_1-k,i_2-k)+(j_2-k,i_2-k)+(j_3-k,i_2-k)\\
&&-(j_0-k,i_3-k)-(j_1-k,i_3-k)+(j_2-k,i_3-k)+(j_3-k,i_3-k).
\end{array}
\end{eqnarray}

\begin{lemma}\label{le reduced}
For each $1\le i,j\le 6$, the correlation of $s_i$ and $s_j$ have the following properties:
\begin{enumerate}
\item [1)] For any $\tau_1,\tau_2\in D_k$, $k=0,1,2,3$, $R_{s_i,s_j}(\tau_1)=R_{s_i,s_j}(\tau_2)$.
\item [2)]
\begin{eqnarray*}
R_{s_i,s_j}(0)&=&\left\{\begin{array}{ll}
n, & i=j,\\
n-2, & i+j=7,\\
1, & \mbox{otherwise}.
\end{array}\right.
\end{eqnarray*}
\item [3)]
For each $\tau\in D_k$ and $l\in D_{k+2}$, where the subscript $k+2$ is performed modulo $4$, we have
\begin{eqnarray*}
R_{s_i,s_j}(\tau)&=&\left\{\begin{array}{ll}
R_{s_j,s_i}(l), & f~\mbox{odd},\\
R_{s_j,s_i}(\tau), & f~\mbox{even}.
\end{array}\right.
\end{eqnarray*}
\end{enumerate}
\end{lemma}

\begin{proof} The proofs of 1) and 2) are obvious, so we only give the proof of 3). Note that $$-1=(-1)^{2f}\in\left\{\begin{array}{ll} D_0, & f~\mbox{odd},\\
D_2, & f~\mbox{even}.\end{array}\right.$$ This implies that $n-\tau\in D_{k+2}$ if $f$ is odd and $n-\tau\in D_k$ if $f$ is even. Hence we have
\begin{eqnarray}\label{eq 2}
(-1)^{s_i(n-\tau)}
=\left\{\begin{array}{ll}(-1)^{s_i(l)}, & f~\mbox{odd},\\
(-1)^{s_i(\tau)}, & f~\mbox{even},\end{array}\right.
~~\Delta_{s_j,s_i}(n-\tau)
&=&\left\{\begin{array}{ll}\Delta_{s_j,s_i}(l), & f~\mbox{odd},\\
\Delta_{s_j,s_i}(\tau), & f~\mbox{even}.\end{array}\right.
\end{eqnarray}
Thus we have
\begin{eqnarray*}
R_{s_i,s_j}(\tau)&=&\sum_{t=0}^{n-1}(-1)^{s_i(t)+s_j(t+\tau)}\\
&=&R_{s_j,s_i}(n-\tau)\\
&=&(-1)^{s_i(n-\tau)}+(-1)^{s_j(\tau)}+\Delta_{s_j,s_i}(n-\tau)\\
&=&\left\{\begin{array}{ll}(-1)^{s_i(l)}+(-1)^{s_j(\tau)}+\Delta_{s_j,s_i}(l), & f~\mbox{odd}\\
(-1)^{s_i(\tau)}+(-1)^{s_j(\tau)}+\Delta_{s_j,s_i}(\tau), & f~\mbox{even}\end{array}\right.\\
&=&\left\{\begin{array}{ll}(-1)^{s_i(l)}+(-1)^{s_j(n-l)}+\Delta_{s_j,s_i}(l), & f~\mbox{odd}\\
(-1)^{s_i(\tau)}+(-1)^{s_j(n-\tau)}+\Delta_{s_j,s_i}(\tau), & f~\mbox{even}\end{array}\right.\\
&=&\left\{\begin{array}{ll}
R_{s_j,s_i}(l), & f~\mbox{odd},\\
R_{s_j,s_i}(\tau), & f~\mbox{even},
\end{array}\right.
\end{eqnarray*}
where the third equal sign is due to (\ref{eqn_cross_st}), and the fourth one is due to (\ref{eq 2}).
\end{proof}

By 3) of Lemma \ref{le reduced}, it is sufficient to consider the correlation of $R_{s_i,s_j}(\tau)$ for $\tau\ne 0$ and $1\le i\le j\le 6$, which are given in the following two theorems.

\begin{theorem}\label{thm odd}
Let $f$ be odd, then the auto- and cross-correlation of $s_1,s_2,s_3,s_4,s_5,s_6$  are given in Table \ref{tab odd f}.

\begin{table}[ht]\caption{The auto- and cross-correlation of six binary sequences of period $n=4f+1$, $f$ odd}\label{tab odd f}
\begin{center}
\begin{tabular}{|c|c|c|c|c|c|}
  \hline
$\tau$ & $\{0\}$ & $D_0$ & $D_1$ & $D_2$ & $D_3$ \\ \hline
$R_{s_1}(\tau)$ & $n$ & $-2y-1$ & $2y-1$ & $-2y-1$ & $2y-1$\\ \hline
$R_{s_2}(\tau)$ & $n$ & $-3$ & $1$ & $-3$ & $1$\\  \hline
$R_{s_3}(\tau)$ & $n$ & $2y-1$ & $-2y-1$ & $2y-1$ & $-2y-1$\\ \hline
$R_{s_4}(\tau)$ & $n$ & $2y-1$ & $-2y-1$ & $2y-1$ & $-2y-1$\\ \hline
$R_{s_5}(\tau)$ & $n$ & $1$ & $-3$ & $1$ & $-3$\\ \hline
$R_{s_6}(\tau)$ & $n$ & $-2y-1$ & $2y-1$ & $-2y-1$ & $2y-1$\\  \hline
$R_{s_1,s_2}(\tau)$ & $1$ & $-x+2y$ & $x+2y+2$ & $x-2y-2$ & $-x-2y$\\  \hline
$R_{s_1,s_3}(\tau)$ & $1$ & $x$ & $-x+2$ & $x$ & $-x-2$\\  \hline
$R_{s_1,s_4}(\tau)$ & $1$ & $-x+2$ & $x$ & $-x-2$ & $x$\\  \hline
$R_{s_1,s_5}(\tau)$ & $1$ & $x-2y+2$  & $-x-2y$ & $-x+2y$ & $x+2y-2$ \\ \hline
$R_{s_1,s_6}(\tau)$ & $2-n$ & $2y+3$  & $3-2y$ & $2y-1$ & $-1-2y$ \\ \hline
$R_{s_2,s_3}(\tau)$ & $1$ & $x+2y-2$  & $x-2y+2$ & $-x-2y$ & $-x+2y$ \\ \hline
$R_{s_2,s_4}(\tau)$ & $1$ & $-x-2y$  & $-x+2y$ & $x+2y-2$ & $x-2y+2$ \\ \hline
$R_{s_2,s_5}(\tau)$ & $2-n$ & $1$  & $1$ & $1$ & $1$ \\ \hline
$R_{s_2,s_6}(\tau)$ & $1$ & $2y-x$  & $x+2y+2$ & $x-2y-2$ & $-x-2y$ \\ \hline
$R_{s_3,s_4}(\tau)$ & $2-n$ & $3-2y$  & $2y-1$ & $-1-2y$ & $3+2y$ \\ \hline
$R_{s_3,s_5}(\tau)$ & $1$ & $x+2y+2$  & $x-2y-2$ & $-x-2y$ & $-x+2y$ \\ \hline
$R_{s_3,s_6}(\tau)$ & $1$ & $-x+2$  & $x$ & $-x-2$ & $x$ \\ \hline
$R_{s_4,s_5}(\tau)$ & $1$ & $-x-2y$  & $-x+2y$ & $x+2y+2$ & $x-2y-2$ \\ \hline
$R_{s_4,s_6}(\tau)$ & $1$ & $x$ & $-x+2$  & $x$ & $-x-2$\\ \hline
$R_{s_5,s_6}(\tau)$ & $1$ & $x-2y+2$  & $-x-2y$ & $-x+2y$ & $x+2y-2$ \\ \hline
\end{tabular}
\end{center}
\end{table}
\end{theorem}

\begin{proof} We only prove the auto-correlation of $s_3$, and the remainder results can be similarly discussed. Let $\tau\in D_k$, $k=0,1,2,3$. By (\ref{eqn_auto_st}), we have
\begin{eqnarray*}
\Delta_{s_3,s_3}(\tau)&=&~~(0-k,0-k)+(3-k,0-k)-(1-k,0-k)-(2-k,0-k)\\
&&+(0-k,3-k)+(3-k,3-k)-(1-k,3-k)-(2-k,3-k)\\
&&-(0-k,1-k)-(3-k,1-k)+(1-k,1-k)+(2-k,1-k)\\
&&-(0-k,2-k)-(3-k,2-k)+(1-k,2-k)+(2-k,2-k)\\
&=&\left\{\begin{array}{ll}
A-3B-C+D+2E, & k=0,2,\\
A+B-C-3D+2E, & k=1,3,
\end{array}\right.\\
&=&\left\{\begin{array}{ll}
A-3B-C+D+2E, & \tau\in D_0\cup D_2,\\
A+B-C-3D+2E, & \tau\in D_1\cup D_3,
\end{array}\right.\\
&=&\left\{\begin{array}{ll}
2y-1, & \tau\in D_0\cup D_2,\\
-2y-1, & \tau\in D_1\cup D_3,
\end{array}\right.
\end{eqnarray*}
where the second equal sign is due to Lemma \ref{le fourorder}.

Note that $f$ is odd, $-1=\alpha^{2f}\in D_2$, and we have $(-1)^{s_3(\tau)}+(-1)^{s_3(n-\tau)}=0$ for any $1\le\tau<n$. By (\ref{eqn_auto_st}) and (\ref{eqn_cross_st}), the auto-correlation of $s_3$ is given as follows
\begin{eqnarray*}
R_{s_3}(\tau)
&=&\left\{\begin{array}{ll}
2y-1, & \tau\in D_0\cup D_2,\\
-2y-1, & \tau\in D_1\cup D_3.
\end{array}\right.
\end{eqnarray*}
\end{proof}

\begin{theorem}\label{thm even}
Let $f$ be even, then the auto- and cross-correlation of $s_1,s_2,s_3,s_4,s_5,s_6$  are given in Table \ref{tab even f}.
\begin{table}[ht]\caption{The auto- and cross-correlation of six binary sequences of period $n=4f+1$, $f$ even}\label{tab even f}
\begin{center}
\begin{tabular}{|c|c|c|c|c|c|}
  \hline
$\tau$ & $\{0\}$ & $D_0$ & $D_1$ & $D_2$ & $D_3$ \\ \hline
$R_{s_1}(\tau)$ & $n$ & $2y-3$ & $-3-2y$ & $1+2y$ & $1-2y$\\ \hline
$R_{s_2}(\tau)$ & $n$ & $-3$ & $1$ & $-3$ & $1$\\  \hline
$R_{s_3}(\tau)$ & $n$ & $-2y-3$ & $2y+1$ & $-2y+1$ & $2y-3$\\ \hline
$R_{s_4}(\tau)$ & $n$ & $-2y+1$ & $2y-3$ & $-2y-3$ & $2y+1$\\ \hline
$R_{s_5}(\tau)$ & $n$ & $1$ & $-3$ & $1$ & $-3$\\ \hline
$R_{s_6}(\tau)$ & $n$ & $2y+1$ & $-2y+1$ & $2y-3$ & $-2y-3$\\  \hline
$R_{s_1,s_2}(\tau)$ & $1$ & $-x+2y-2$ & $x+2y$ & $x-2y$ & $-x-2y+2$\\  \hline
$R_{s_1,s_3}(\tau)$ & $1$ & $-x-2$ & $x$ & $-x+2$ & $x$\\  \hline
$R_{s_1,s_4}(\tau)$ & $1$ & $x$ & $-x-2$ & $x$ & $-x+2$\\  \hline
$R_{s_1,s_5}(\tau)$ & $1$ & $x-2y$  & $-x-2y-2$ & $-x+2y+2$ & $x+2y$ \\ \hline
$R_{s_1,s_6}(\tau)$ & $2-n$ & $1-2y$  & $1+2y$ & $1-2y$ & $1+2y$ \\ \hline
$R_{s_2,s_3}(\tau)$ & $1$ & $-x-2y-2$  & $-x+2y+2$ & $x+2y$ & $x-2y$ \\ \hline
$R_{s_2,s_4}(\tau)$ & $1$ & $x+2y$  & $x-2y$ & $-x-2y-2$ & $-x+2y+2$ \\ \hline
$R_{s_2,s_5}(\tau)$ & $2-n$ & $1$  & $1$ & $1$ & $1$ \\ \hline
$R_{s_2,s_6}(\tau)$ & $1$ & $x-2y$  & $-x-2y+2$ & $-x+2y-2$ & $x+2y$ \\ \hline
$R_{s_3,s_4}(\tau)$ & $2-n$ & $1+2y$  & $1-2y$ & $1+2y$ & $1-2y$ \\ \hline
$R_{s_3,s_5}(\tau)$ & $1$ & $x+2y$  & $x-2y$ & $-x-2y+2$ & $-x+2y-2$ \\ \hline
$R_{s_3,s_6}(\tau)$ & $1$ & $x$  & $-x+2$ & $x$ & $-x-2$ \\ \hline
$R_{s_4,s_5}(\tau)$ & $1$ & $-x-2y+2$  & $-x+2y-2$ & $x+2y$ & $x-2y$ \\ \hline
$R_{s_4,s_6}(\tau)$ & $1$ & $-x+2$  & $x$ & $-x-2$ & $x$ \\ \hline
$R_{s_5,s_6}(\tau)$ & $1$ & $-x+2y+2$  & $x+2y$ & $x-2y$ & $-x-2y-2 $ \\ \hline
\end{tabular}
\end{center}
\end{table}
\end{theorem}

\begin{proof} We only prove the auto-correlation of $s_3$, and the remainder results can be similarly discussed. By (\ref{eqn_auto_st}), we have
\begin{eqnarray*}
\Delta_{s_3,s_3}(\tau)&=&(0-k,0-k)+(0-k,3-k)-(0-k,1-k)-(0-k,2-k)\\
&&+(3-k,0-k)+(3-k,3-k)-(3-k,1-k)-(3-k,2-k)\\
&&-(1-k,0-k)-(1-k,3-k)+(1-k,1-k)+(1-k,2-k)\\
&&-(2-k,0-k)-(2-k,3-k)+(2-k,1-k)+(2-k,2-k)\\
&=&\left\{\begin{array}{ll}
A-B-C+3D-2E, & k=0,2,\\
A+3B-C-D-2E, & k=1,3,\\
\end{array}\right.\\
&=&\left\{\begin{array}{ll}
-1-2y, & k=0,2,\\
-1+2y, & k=1,3,
\end{array}\right.\\
&=&\left\{\begin{array}{ll}
-1-2y, & \tau\in D_0\cup D_2,\\
-1+2y, & \tau\in D_1\cup D_3,
\end{array}\right.
\end{eqnarray*}
where the second equal sign is due to Lemma \ref{le fourorder}.

Note that $f$ is even, we have $-1=\alpha^{2f}\in D_0$, and then $n-\tau\in D_k$ for $\tau\in D_k$. Hence one has
$$(-1)^{s_3(\tau)}+(-1)^{s_3(n-\tau)}=\left\{\begin{array}{ll}
-2, &\tau\in D_0\cup D_3,\\
2, & \tau\in D_1\cup D_2. \end{array}\right.$$
By (\ref{eqn_auto_st}) and (\ref{eqn_cross_st}), the auto-correlation of $s_3$ is given as follows
\begin{eqnarray*}
R_{s_3}(\tau)
&=&\left\{\begin{array}{ll}
-3-2y, & \tau\in D_0,\\\
1+2y, & \tau\in D_1,\\
1-2y, & \tau\in D_2,\\
-3+2y, & \tau\in D_3.
\end{array}\right.
\end{eqnarray*}
\end{proof}

\end{appendix}

\end{document}